\documentclass[aps,twocolumn,10pt,superscriptaddress]{revtex4-2}

\usepackage{mathrsfs}
\usepackage{bbm}
\usepackage{amsmath}
\usepackage{amssymb}
\usepackage{amsthm}
\usepackage{graphicx}
\usepackage{MnSymbol}
\usepackage{mathtools}
\usepackage{stmaryrd}
\usepackage{enumitem}
\usepackage{bbold}
\usepackage[italicdiff]{physics}
\usepackage{dsfont}

\makeatletter
\usepackage{hyperref}
\usepackage{color}
\definecolor{supcol}{RGB}{10,50,180}
\definecolor{eqcol}{RGB}{220,10,100}
\hypersetup{
	colorlinks,
	citecolor=supcol,
	linkcolor=eqcol,
	urlcolor=supcol
}

\allowdisplaybreaks

\newtheorem{theorem}{Theorem}

\newtheorem{lemma}[theorem]{Lemma}

\newcommand{\mca}{\mathcal}
\newcommand{\mbb}{\mathbb}
\newcommand{\mds}{\mathds}

\newcommand{\sectionprl}[1]{{\em #1}\/.---}

\begin{document}
\title{Fidelity-dissipation relations in quantum gates}

\author{Tan Van Vu}
\email{tan.vu@yukawa.kyoto-u.ac.jp}
\affiliation{Analytical quantum complexity RIKEN Hakubi Research Team, RIKEN Center for Quantum Computing (RQC), 2-1 Hirosawa, Wako, Saitama 351-0198, Japan}
\affiliation{Yukawa Institute for Theoretical Physics, Kyoto University, Kitashirakawa Oiwakecho, Sakyo-ku, Kyoto 606-8502, Japan}

\author{Tomotaka Kuwahara}
\email{tomotaka.kuwahara@riken.jp}
\affiliation{Analytical quantum complexity RIKEN Hakubi Research Team, RIKEN Center for Quantum Computing (RQC), 2-1 Hirosawa, Wako, Saitama 351-0198, Japan}
\affiliation{PRESTO, Japan Science and Technology (JST), Kawaguchi, Saitama 332-0012, Japan}

\author{Keiji Saito}
\email{keiji.saitoh@scphys.kyoto-u.ac.jp}
\affiliation{Department of Physics, Kyoto University, Kyoto 606-8502, Japan}

\date{\today}

\begin{abstract}
Accurate quantum computing relies on the precision of quantum gates. However, quantum gates in practice are generally affected by dissipative environments, which can significantly reduce their fidelity. In this study, we elucidate fundamental relations between the average fidelity of generic quantum gates and the dissipation that occurs during the computing processes. Considering scenarios in which a quantum gate is subject to Markovian environments, we rigorously derive fidelity-dissipation relations that hold for arbitrary operational times. Intriguingly, when the quantum gate undergoes thermal relaxation, the result can be used as a valuable tool for estimating dissipation through experimentally measurable fidelity, without requiring detailed knowledge of the dissipative structure. For the case of arbitrary environments, we uncover a trade-off relation between the average fidelity and energy dissipation, implying that these quantities cannot be large simultaneously. Our results unveil the computational limitations imposed by thermodynamics, shedding light on the profound connection between thermodynamics and quantum computing.
\end{abstract}

\maketitle

\section{Introduction}
Quantum computing, which exploits quantum features like coherence and entanglement, has the potential to outperform classical computers in terms of computational power \cite{Nielsen.2000,Ladd.2010.N,Harrow.2017.N,Arute.2019.N}. 
Over the past few decades, substantial efforts have been devoted to the experimental development of quantum computers using various platforms, such as superconducting qubits \cite{Wendin.2017.RPP,Gu.2017.PR,Kjaergaard.2020.ARCMP,Blais.2021.RMP}, trapped ions \cite{Haffner.2008.PR,Bruzewicz.2019.APR}, photonic systems \cite{Takeda.2019.APLP,Zhong.2020.S,Madsen.2022.N}, and nuclear magnetic resonance \cite{Warren.1997.S,Cory.1997.PNAS,Jones.2011.PNMRS}.
To achieve the long-sought goal of building a reliable quantum computer, numerous advancements are essential from both theoretical and experimental perspectives. 
A critical aspect of this endeavor involves developing a profound comprehension of the energetics of quantum computing \cite{Landauer.1961.JRD,Banacloche.2002.PRL,Sagawa.2009.PRL,Reeb.2014.NJP,Faist.2015.NC,Ikonen.2017.npjQI,Tajima.2018.PRL,Landi.2020.PRA,Cimini.2020.npjQI,Chiribella.2021.PRX,Deffner.2021.EPL,Funo.2021.PRL,Danageozian.2022.PRXQ,Stevens.2022.PRL,Vu.2022.PRL,Auffeves.2022.PRXQ}, an area that intimately intersects with related fields such as stochastic and quantum thermodynamics \cite{Sekimoto.2010,Seifert.2012.RPP,Vinjanampathy.2016.CP,Goold.2016.JPA,Deffner.2019}.

\begin{figure}[b]
\centering
\includegraphics[width=1\linewidth]{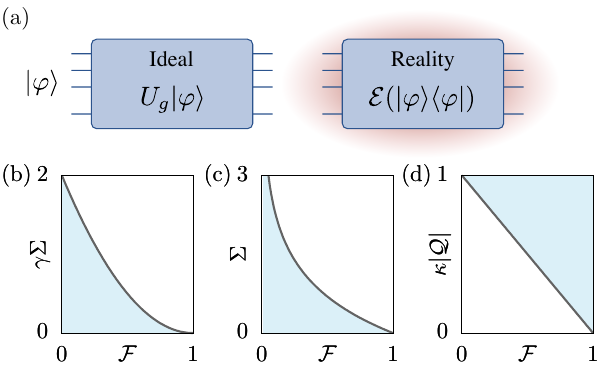}
\protect\caption{(a) The top panel is the schematic of a generic quantum gate acting on qubits. 
Ideally, the quantum gate is characterized by a unitary operator $U_g$ that transforms an input pure state $\ket{\varphi}$ into $U_g\ket{\varphi}$. 
In reality, however, the dynamics of the quantum gate is affected by the dissipative environment and becomes nonunitary, which is characterized by a completely-positive trace preserving map $\mca{E}(\cdot)$. 
The bottom panel depicts fundamental relations between fidelity and dissipation in the quantum gate that undergoes (b) an arbitrary time-dependent driving, (c) a thermal relaxation process, and (d) a general thermodynamic process. 
The light blue areas indicate infeasible regions wherein these quantities cannot be attained simultaneously.}\label{fig:Cover}
\end{figure}

Quantum gate operations are basic components in quantum information processing. 
In the pursuit of quantum computing, the availability of high-fidelity quantum gates is indispensable. 
However, it is important to acknowledge that quantum gates, by their nature, cannot be perfectly shielded from interactions with their dissipative environment \cite{Knill.2005.N}. 
Consequently, quantum decoherence inevitably arises, resulting in a degradation of fidelity. 
To tackle this challenge, numerous engineering proposals have been reported in recent years, aiming at implementing quantum gates with high speed and fidelity \cite{Gaebler.2016.PRL,Ballance.2016.PRL,Huang.2019.PRL,Hegde.2022.PRL,Erickson.2022.PRL,Xie.2023.PRL,Leu.2023.PRL}. 
Nevertheless, it remains crucial, at a fundamental level, to understand the impact of environment-induced decoherence on quantum gate fidelity from a thermodynamic standpoint. 
This direction bears significance not only for the ultimate objective of establishing fault-tolerant quantum computing but also for elucidating the energetic costs associated with quantum computing. 
To date, several studies have explored gate fidelity under specific scenarios such as in the presence of conservative laws \cite{Ozawa.2002.PRL}, short operational times \cite{Abad.2022.PRL}, imperfect timekeeping \cite{Xuereb.2023.PRL}, and fluctuations of external control \cite{Jiang.2023.PRA}. 
Nevertheless, the intimate relationship between fidelity and dissipation that occurs during computation remains veiled. 
To achieve a deeper comprehension of quantum computing, it is relevant to quantitatively uncover thermodynamic effects on gate fidelity.

In the present study, we utilize the framework of quantum thermodynamics to establish fundamental relations between the average fidelity and dissipation inherent to generic quantum gates (see Fig.~\ref{fig:Cover} for illustration). 
Assuming that the dissipative dynamics of quantum gates is governed by the Markovian master equation, we prove that the sum of the average fidelity and the square root of dissipation always equals or exceeds one [cf.~Eq.~\eqref{eq:main.res}]. 
This finding qualitatively implies that these quantities cannot simultaneously be small, and this holds true for arbitrary time-dependent driving and operational durations. 
Notably, in scenarios where the quantum gate undergoes thermal relaxation (i.e., when the Hamiltonian is time-independent), we derive that the product of the average fidelity and the exponential of dissipation is consistently greater than or equal to one [cf.~Eq.~\eqref{eq:main.res.therm.relax}]. 
Significantly, this relation remains valid regardless of the specific details of the dissipative structure, rendering it applicable for inferring dissipation through experimental measurements of the average fidelity. 
Beyond Markovian cases, we reveal a trade-off relation between gate fidelity and energy dissipation for general quantum gates coupled to an arbitrary environment [cf.~Eq.~\eqref{eq:FQ.tradeoff}]. 
This trade-off indicates that these two quantities cannot be large simultaneously. 
These findings yield sandwich bounds on gate fidelity, providing a comprehensive characterization of thermodynamic effects on quantum gates.

\section{Setup}
We consider a generic quantum gate acting on $N$ qubits. 
In the absence of interactions with the environment, the quantum gate can be ideally described by a unitary operator $U_g$ with the underlying Hamiltonian $H_t$; that is, $U_g=\vec{\mbb{T}}\exp\qty(-i\int_0^\tau\dd{t}H_t)$, where $\tau$ denotes the operational time and $\vec{\mbb{T}}$ denotes the time-ordering operator.
However, in reality, any physical system cannot be perfectly isolated from the environment, thus basically being an open system.
The unavoidable interactions with the dissipative environment act as noises and may degrade the fidelity of quantum gates.
Let $\mca{E}(\cdot)$ denote the completely positive trace-preserving map that represents the actual quantum gate. That is, $\varrho_\tau=\mca{E}(\varrho_0)$, where $\varrho_t$ is the density matrix of the quantum gate at time $t$.
We assume that the quantum gate is weakly coupled to the memoryless environment.
In this case, the time evolution of $\varrho_t$ can be described by the Gorini-Kossakowski-Sudarshan-Lindblad (GKSL) equation \cite{Lindblad.1976.CMP,Gorini.1976.JMP}:
\begin{equation}\label{eq:Lindblad.eq}
\dot\varrho_t=-i[H_t,\varrho_t]+\sum_c\mca{D}[L_c(t)]\varrho_t,
\end{equation}
where $\mca{D}[L]\varrho\coloneqq L\varrho L^\dagger-\qty{L^\dagger L,\varrho}/2$ is the dissipator and $\{L_c\}_c$ denote the jump operators, which characterize the decoherence effect of the environment on the system.
The jump operators, in general, can be classified into two types: dissipative and non-dissipative.
Examples of the first type of operators include relaxation and excitation, which occur due to energy exchange with the environment.
Such dissipative jumps significantly contribute to the irreversibility of the computing process.
The second one encompasses phase damping, which causes the loss of quantum information without the energy loss.
Throughout this study, we consider \emph{both} types of dissipative and non-dissipative jump operators.
To guarantee the thermodynamical consistency of the quantum gate, we assume the local detailed balance for dissipative jump operators \cite{Horowitz.2013.NJP,Manzano.2018.PRX}; that is, they come in pairs $(c,c')$ such that 
\begin{equation}\label{eq:ldbc}
	L_c(t)=e^{s_c(t)/2}L_{c'}(t)^\dagger,
\end{equation}
where $s_c(t)=-s_{c'}(t)$ denotes the environmental entropy change associated with the $c$th jump.
In addition, non-dissipative operators are assumed to be Hermitian, which can also be cast to Eq.~\eqref{eq:ldbc} by setting $c'=c$ and $s_c(t)=0$.
Note that the jump operators are not specified and can have arbitrary forms.
Various noises, such as phase damping, amplitude damping, and crosstalk between qubits, can be included into these two classes of jump operators.

The fidelity of the actual quantum gate is an indicator of significant interest.
For each initial pure state $\varphi=\dyad{\varphi}$, the output states of the ideal and actual quantum gates are $U_g{\varphi}U_g^\dagger$ and $\mca{E}(\varphi)$, respectively.
Using the fidelity between these output states, the average gate fidelity of the quantum operation is given by \cite{Nielsen.2002.PLA}
\begin{equation}
\mca{F}\coloneqq\int\dd{\varphi}\mel{\varphi}{U_g^\dagger\mca{E}(\varphi)U_g}{\varphi},
\end{equation}
where the integral is over all pure states and $d\varphi$ denotes the normalized Haar measure.
The quantity $\mca{F}$ quantifies how accurately the channel $\mca{E}$ approximates the quantum gate $U_g$.
When the map $\mca{E}$ perfectly implements the gate operation, $\mca{F}=1$. In general, we have $\mca{F}\le 1$.

Another relevant quantity is dissipation, which reflects the irreversible nature of the quantum gate.
Using the framework of quantum thermodynamics, dissipation in open quantum systems can be quantified by irreversible entropy production \cite{Landi.2021.RMP}.
For the initial state $\varrho_0={\varphi}$, the associated entropy production can be defined as the sum of entropy changes in the system and environment as
\begin{align}
\Sigma_\varphi(\tau)\coloneqq \Sigma_\varphi^{\rm sys}(\tau) + \Sigma_\varphi^{\rm env}(\tau).
\end{align}
Here, the first term $\Sigma_\varphi^{\rm sys}(\tau)\coloneqq S(\varrho_\tau)-S(\varrho_0)$ is the system's entropy production quantified by the von Neumann entropy $S(\varrho)\coloneqq-\tr{\varrho\ln\varrho}$ and the second term $\Sigma_\varphi^{\rm env}(\tau)\coloneqq\int_0^\tau\dd{t}\sum_{c}\tr{L_c(t)\varrho_tL_c(t)^\dagger}s_c(t)$ characterizes the amount of heat dissipated into the environment.
We can prove that $\Sigma_\varphi(\tau)\ge 0$ for arbitrary initial states, which corresponds to the second law of thermodynamics.
In analogy to the gate fidelity, the average dissipation over all pure states can be defined as
\begin{equation}
\Sigma\coloneqq\int\dd{\varphi}\Sigma_\varphi(\tau).
\end{equation}
In this work, we aim to explore fundamental relations between fidelity $\mca{F}$ and dissipation terms such as $\Sigma$ for arbitrary operational time $\tau$.

\section{Main results}
Within the above setup, we are now going to explain our results, which unveil the intimate relationships between the fidelity and dissipation of the quantum gate.
First, we prove that for an arbitrary control protocol, the average gate fidelity and the average dissipation always obey the following relation:
\begin{equation}\label{eq:main.res}
\mca{F}+\sqrt{\gamma\Sigma/2}\ge 1.
\end{equation}
Here, $\gamma\coloneqq\int_0^\tau\dd{t}\sum_{c}[\Delta L_c(t)]^2$ is a quantity that is determined solely by the jump operators, where $[\Delta L]^2\coloneqq(d\tr{L^\dagger L} - |\tr L|^2)/d(d+1)\ge 0$ and $d=2^N$ being the dimension of Hilbert space.
The inequality \eqref{eq:main.res} is our first main result, which holds true for arbitrary operational times.
Its proof is presented in Appendix \ref{app:proof.res.1}.
While qubits quantum gates are considered here, the generalization to qudits systems is straightforward.

Some remarks on the relation \eqref{eq:main.res} are in order.
First, it expresses a quantitative relationship between the average fidelity $\mca{F}$ and the thermodynamic cost $\Sigma$, providing insights into the extent to which dissipation can impact the gate fidelity.
Furthermore, it can also be regarded as a thermodynamic upper bound on the gate error.
Second, this relation applies to arbitrary quantum gates as it does not explicitly depend on the form of the functional unitary $U_g$ but only jump operators $\{L_c\}$.
Third, a practical advantage of this bound is that it allows for the estimation of dissipation by examining the rates of decoherence and dephasing.
To demonstrate this, let us consider a typical example, in which the operational time $\tau$ is significantly short (i.e., $\tau\ll 1$) and the quantum gate is affected by Pauli noises (i.e., $L_c\in\sqrt{\Gamma_c}\{\mds{1},\sigma_x,\sigma_y,\sigma_z\}^{\otimes N}$, where $\tr L_c=0$ and $\Gamma_c$ denotes the jump rate).
In this case, the average fidelity can be calculated as $\mca{F}=1-\gamma+\mca{O}(\Gamma^2\tau^2)$ and $\gamma=\tau \Gamma d/(d+1)$, where $\Gamma\coloneqq\sum_c\Gamma_c$ is the total dephasing rate.
By applying the relation \eqref{eq:main.res}, a simple lower bound on dissipation can be obtained as
\begin{equation}
	\Sigma\ge\frac{2d}{d+1}\Gamma\tau+\mca{O}(\Gamma^2\tau^2).
\end{equation}
Note that the decay rate $\Gamma$ can be, in principle, estimated in experiments \cite{Harper.2020.NP,Harper.2021.PRXQ,Flammia.2021.Q}.
It is worth discussing relevant \emph{indirect} methods for inferring dissipation. In classical systems, a method based on the thermodynamic uncertainty relation is often employed \cite{Barato.2015.PRL,Koyuk.2020.PRL,Manikandan.2020.PRL,Vu.2020.PRE}. 
For open quantum systems, the findings presented in Ref.~\cite{Vu.2022.PRL.TUR} can establish a lower bound for entropy production, which is expressed as $\Sigma\ge 2\int\dd{\varphi}(\Psi_J-Q)$.
Here, $\Psi_J$ represents a statistical value of current $J$ defined over a stochastic trajectory and $Q$ denotes a quantum contribution. 
However, accurately computing $Q$ requires a detailed understanding of the underlying dynamics, which may be challenging to achieve in practical experiments \cite{fnt1}. 
Additionally, this lower bound can become negative if the quantum contribution exceeds the statistical value of the current.
Fourth, while the Hamiltonian $H_t$ is implicitly assumed to be perfectly implemented without any errors, the result can be readily generalized to encompass the case of imperfect implementations.
In practice, a desired Hamiltonian may not be exactly built due to the limitations of controls or systematic errors.
For example, they include quantum crosstalk noises such as local $ZZ$ interactions between qubits in the Hamiltonian \cite{Zhou.2023.PRL}.
Given the implemented Hamiltonian $\hat H_t$, which possibly different from the ideal $H_t$, we can obtain the following generalized relation (see Appendix \ref{app:proof.res.gen} for the proof):
\begin{equation}\label{eq:main.res.gen}
	\mca{F}+\sqrt{\gamma\Sigma/2}\ge\frac{|\tr{\hat U_\tau^\dagger U_g}|^2+d}{d(d+1)},
\end{equation}
where $\hat U_\tau\coloneqq\vec{\mbb{T}}\exp(-i\int_0^\tau\dd{t}\hat H_t)$.
For the perfect case (i.e., when $\hat H_t=H_t$), $\tr{\hat U_\tau^\dagger U_g}=d$ and the relation \eqref{eq:main.res} is immediately recovered.
Last, we emphasize that the local detailed balance condition is the minimal requirement for discussing the thermodynamics of fidelity. When this condition is excluded, a trade-off relation between fidelity and activity (which is not a thermodynamic quantity) can still be derived (see Appendix \ref{app:proof.fa} for details).

Next, we consider the case where the quantum gate undergoes a thermal relaxation process.
That is, both the underlying Hamiltonian and the jump operators are time-independent, and each jump operator accounts for a transition between energy eigenstates of the Hamiltonian.
It is noteworthy that constructing a unitary gate with a time-independent Hamiltonian provides several advantages, which can simultaneously enhance computational speed and accuracy. 
Specifically, this approach not only effectively reduces noise from external control but also ensures a more reliable implementation by eliminating concerns about the complex dynamics required to implement a given gate \cite{Innocenti.2020.NJP}. 
For this relevant class of quantum gates, we rigorously obtain the following inequality:
\begin{equation}\label{eq:main.res.therm.relax}
	\mca{F}e^{\Sigma}\ge 1,
\end{equation}
whose proof is presented in Appendix \ref{app:proof.res.ther.rel}.
While the applicable range of this compact relation is more limited compared to the relation \eqref{eq:main.res}, it has its own strength.
That is, it involves only fidelity and dissipation, and is independent of the specific details of the Hamiltonian and jump operators.
Consequently, it offers a promising approach for estimating the dissipation of quantum gates using experimentally measurable fidelity.

Thus far, we have uncovered thermodynamic constraints for quantum gates, establishing lower bounds on gate fidelity.
In what follows, we extend our analysis to encompass thermodynamic upper bounds on gate fidelity in a broader context.

\sectionprl{Trade-off between fidelity and energy dissipation}We consider a general case beyond Markovian environments.
That is, the quantum gate is coupled to an arbitrary environment, and the composite system evolves under a unitary operator.
Specifically, the actual quantum gate is characterized by the map $\mca{E}(\cdot)=\tr_E\{U(\cdot\otimes\varrho_E)U^\dagger\}$, where $U$ is a unitary operator and $\varrho_E$ is an arbitrary initial state of the environment.
We investigate the case of a time-independent Hamiltonian (i.e., $H_t=H$ for all $t$) and arbitrary interactions with the environment; the generalization to the time-dependent case is straightforward (Appendix \ref{app:proof.ener.diss}).
In this case, the average energy change in the quantum gate can be defined as
\begin{equation}
	\mca{Q}\coloneqq\int\dd{\varphi}\tr{H(\varrho_\tau-\varrho_0)}.
\end{equation}
When the unitary operator $U$ conserves the total energy, $\mca{Q}$ is equivalent to the average heat dissipated into the environment.
Notice that $\mca{Q}$ vanishes if the gate is isolated from the environment.
For this general setting, we find a trade-off relation between the gate fidelity and the thermodynamic quantity $\mca{Q}$, given by (see Appendix \ref{app:proof.ener.diss} for the proof)
\begin{equation}\label{eq:FQ.tradeoff}
	\mca{F}+\kappa|\mca{Q}|\le 1,
\end{equation}
where $\kappa\coloneqq{d}/{(d+1)g}$ and $g$ denotes the energy bandwidth of the gate Hamiltonian $H$.
This result provides a quantitative implication that energy dissipation must be suppressed to achieve high fidelity.
Additionally, the relation \eqref{eq:FQ.tradeoff} can be interpreted either as an upper bound on gate fidelity in terms of energy dissipation or as an upper bound on energy dissipation in terms of gate fidelity.

\section{Numerical demonstration}
We exemplify our results in single-qubit and controlled-$Z$ (CZ) quantum gates, which are universal for quantum computation \cite{Nielsen.2000} and have been experimentally implemented on various platforms.

\begin{figure}[t]
\centering
\includegraphics[width=1\linewidth]{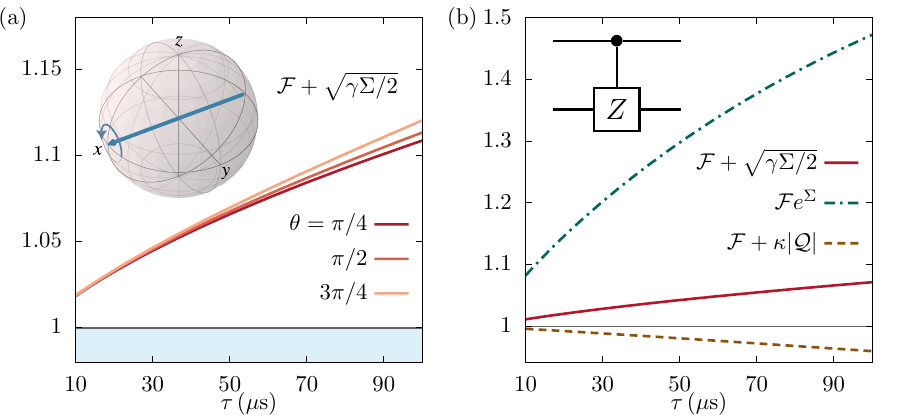}
\protect\caption{Numerical illustration of the relations \eqref{eq:main.res}, \eqref{eq:main.res.therm.relax}, and \eqref{eq:FQ.tradeoff} in (a) the $X_\theta$ gate and (b) the CZ gate. Time $\tau$ is varied in range $[10,100]\,\mu$s while the other parameters are set as $\Omega=5$\,GHz, $T=20$\,mK, and $\Gamma_a=10^{-6}$. The constant rate $\Gamma_p$ is determined such that the longitudinal relaxation time $T_1$ and the transverse relaxation time $T_2$ satisfy $T_1=\sqrt{2}T_2$.}\label{fig:Res}
\end{figure}

We first demonstrate the result \eqref{eq:main.res} using single-qubit gates.
The Hamiltonian is given by $H_t=\omega_t(\sigma_x\cos\phi_t+\sigma_y\sin\phi_t)/2$, where $\omega_t$ and $\phi_t$ are control parameters.
Through proper manipulations of $\omega_t$ and $\phi_t$, rotations about the $x$- and $y$-axis can be implemented, which are sufficient to build any single-qubit unitary operator up to a global phase \cite{Nielsen.2000}.
We particularly examine the $X_\theta$ gate, which rotates the qubit state by $\theta$ radians about the $x$-axis.
This gate can be implemented by the following time-dependent protocol: $\omega_t=(\theta/2\tau)[1+2\sin(\pi t/2\tau)^2]$ and $\phi_t=0$.
Assume that the quantum gate is weakly coupled to the environment at inverse temperature $\beta$ and subject to the thermal noise characterized by two jump operators $L_1=\sqrt{\Gamma_a\Omega(\bar{n}+1)}\sigma_-$ and $L_2=\sqrt{\Gamma_a\Omega\bar{n}}\sigma_+$.
Here, $\Gamma_a$ and $\Omega$ denote the coupling strength and the qubit energy gap, respectively, and $\bar{n}=1/(e^{\hbar\beta\Omega}-1)$.
In addition, the quantum gate is also affected by the dephasing noise, characterized by operator $L_3=\sqrt{\Gamma_p/2}\sigma_z$.
For each $\theta\in\{\pi/4,\pi/2,3\pi/4\}$, we vary the operational time $\tau$ and calculate the average fidelity and dissipation for each instant of time.
As seen in Fig.~\ref{fig:Res}(a), the lines representing the fidelity-dissipation relation \eqref{eq:main.res} lie above the constant line of one, thereby numerically confirming the validity of the bound.

Next, we examine the two-qubit CZ gate, which can be realized using the time-independent Hamiltonian $H=(\omega/2)(\sigma_z\otimes\mds{1}+\mds{1}\otimes\sigma_z-\sigma_z\otimes\sigma_z)$, where $\omega=\pi/2\tau$ \cite{Krantz.2019.APR}.
Besides the local dephasing, we also assume that the gate is influenced by correlated noises, characterized by two jump operators $L_1\propto\sigma_-\otimes\sigma_-$ and $L_2\propto\sigma_+\otimes\sigma_+$.
In this case, the quantum gate simply undergoes a thermal relaxation process; therefore, all the relations \eqref{eq:main.res}, \eqref{eq:main.res.therm.relax}, and \eqref{eq:FQ.tradeoff} are applicable.
We analogously alter the operational time and plot the relevant quantities in Fig.~\ref{fig:Res}(b).
It can be numerically verified that the derived relations hold true for all times and can be saturated in the $\tau\to 0$ limit.

\section{Summary and outlook}
In this study, we established the fundamental relations between the average fidelity and dissipation for generic quantum gates. 
These relations serve a dual purpose: they quantify the inherent trade-off between fidelity and dissipation while also offering valuable tools for estimating the thermodynamic cost of quantum computing. 
Furthermore, they quantitatively reveal the performance limitations caused by dissipation and provide guidance for designing optimal quantum circuits from a thermodynamic standpoint. 
Since our relations hold true at the level of a single initial pure state, they can be easily extended to other scenarios. 
For instance, the effect of imperfect timekeeping \cite{Xuereb.2023.PRL} can be incorporated by considering ensemble averages over the distribution of operational times.

We discuss several possible extensions and research directions that our study can motivate. 
One feasible extension involves incorporating the readout process into quantum circuits \cite{Wiseman.2009}. 
Since such processes cause energy exchange between the circuit and controller, this investigation would provide a more comprehensive understanding of the relationship between dissipation and computation reliability. 
Beyond examining the impact of dissipation on fidelity, it is essential to consider other thermodynamic aspects. 
One critical aspect of quantum circuits is information scrambling \cite{Nahum.2018.PRX,Landsman.2019.N,Mi.2021.S}, which quantifies the rate at which local information spreads. 
This concept has received extensive attention in various contexts, such as thermalization \cite{Rigol.2008.N}, chaos \cite{Maldacena.2016.JHEP}, many-body localization \cite{Fan.2017.SB}, and black-hole physics \cite{Shenker.2014.JHEP}. 
Recently, numerical studies have shown that information spreading can be suffered by dissipation \cite{Zhang.2019.PRB}. 
Therefore, elucidating analytical relations between dissipation and spreading measures, like the out-of-time-order correlators, is of fundamental importance. 
The techniques developed in our study may prove helpful in exploring this connection. 
Another intriguing avenue of research is exploring the interplay between thermodynamics and complexity in quantum circuits. 
Quantum Wasserstein distances, which have been demonstrated to play critical roles in both of these domains \cite{Vu.2021.PRL,DePalma.2021.TIT,Li.2022.arxiv,Vu.2023.PRL.TSL,Vu.2023.PRX}, could potentially offer an approach to address this matter.

\begin{acknowledgments}
The authors thank Eisuke Abe and Yutaka Tabuchi for the fruitful discussion.
{T.V.V.}~was supported by JSPS KAKENHI Grant No.~JP23K13032.
{T.K.}~acknowledges Hakubi projects of RIKEN and was supported by JST, PRESTO Grant No.~JPMJPR2116.
{K.S.}~was supported by JSPS KAKENHI Grant No.~JP23K25796.
\end{acknowledgments}

\appendix

\section{Proof of the fidelity-dissipation relation \eqref{eq:main.res}}\label{app:proof.res.1}

\subsection{Interaction picture}
We consider the interaction picture to derive the relation.
Define $U_t\coloneqq\vec{\mbb{T}}\exp\qty(-i\int_0^t\dd{s}H_s)$, where $\vec{\mbb{T}}$ denotes the time-ordering operator.
This time-dependent operator satisfies the following differential equation:
\begin{equation}
\dot{U}_t=-iH_tU_t,
\end{equation}
where $U_0=\mds{1}$.
For each time-dependent operator $X_t$, we also define the corresponding operator in the interaction picture, $\tilde{X}_t\coloneqq U_t^\dagger X_tU_t$.

In what follows, we show that the density matrix $\tilde{\varrho}_t$ obeys the following GKSL master equation \cite{Vu.2021.PRL}:
\begin{equation}
\dot{\tilde{\varrho}}_t=\sum_c\mca{D}[\tilde{L}_c(t)]\tilde{\varrho}_t.\label{eq:Lind.eq.int.pic}
\end{equation}
Taking the time derivative of $\tilde{\varrho}_t$, Eq.~\eqref{eq:Lind.eq.int.pic} can be obtained as
\begin{equation}
\begin{aligned}[b]
\dot{\tilde{\varrho}}_t&=U_t^\dagger i[H_t,\varrho_t]U_t+U_t^\dagger\dot{\varrho}_tU_t\\
&=U_t^\dagger\sum_c\qty[L_c(t)\varrho_tL_c(t)^\dagger-\frac{1}{2}\{L_c(t)^\dagger L_c(t),\varrho_t\}]U_t\\
&=\sum_c\tilde{L}_c(t)\tilde{\varrho}_t\tilde{L}_c(t)^\dagger-\frac{1}{2}\{\tilde{L}_c(t)^\dagger \tilde{L}_c(t),\tilde{\varrho}_t\}\\
&=\sum_c\mca{D}[\tilde{L}_c(t)]\tilde{\varrho}_t.
\end{aligned}
\end{equation}
It should be noted that the dissipative jump operators in the interaction picture satisfy the local detailed balance condition:
\begin{equation}
	\tilde{L}_c(t)=e^{s_c(t)/2}\tilde{L}_{c'}(t)^\dagger.
\end{equation}

\subsection{Entropy production rate}
Next, we derive an analytical expression of the entropy production rate $\dot\Sigma_{\varphi}(t)\coloneqq (d/dt)\Sigma_{\varphi}(t)$.
Hereafter, we omit the time notation for simplicity.
Taking the time derivative of irreversible entropy production, we can calculate the entropy production rate as
\begin{align}
\dot\Sigma_{\varphi}&=-\tr{\dot\varrho\ln\varrho}+\sum_c\tr{L_c^\dagger L_c\varrho}s_c\notag\\
&=-\tr{\dot{\tilde{\varrho}}\ln\tilde{\varrho}}+\sum_c\tr{\tilde{L}_c^\dagger\tilde{L}_c\tilde{\varrho}}s_c\notag\\
&=-\sum_c\tr{(\mca{D}[\tilde{L}_c]\tilde{\varrho})\ln\tilde{\varrho}}+\tr{\tilde{L}_c^\dagger\tilde{L}_c\tilde{\varrho}}s_c\notag\\
&=\sum_c\tr{\tilde{L}_c\tilde{\varrho}(s_c\tilde{L}_c^\dagger-[\tilde{L}_c^\dagger,\ln\tilde{\varrho}])}.
\end{align}
Let $\tilde{\varrho}=\sum_np_n\dyad{n}$ be the spectral decomposition.
Define transition rates $w_{mn}^c\coloneqq|\mel{m}{\tilde{L}_c}{n}|^2$, probability currents $j_{mn}^c\coloneqq w_{mn}^cp_n-w_{nm}^{c'}p_m$, and thermodynamic forces $f_{mn}^c\coloneqq{\rm ln}(w_{mn}^cp_n/w_{nm}^{c'}p_m)$.
Notice that $w_{mn}^c=e^{s_c}w_{nm}^{c'}$.
Since $\tr A=\sum_m\mel{m}{A}{m}$ for any operator $A$, the entropy production rate can be calculated further as \cite{Vu.2024.PRA}
\begin{align}
\dot\Sigma_{\varphi}&=\sum_{m,c}\mel{m}{\tilde{L}_c\tilde{\varrho}( s_c\tilde{L}_c^\dagger-[\tilde{L}_c^\dagger,\ln\tilde{\varrho}])}{m}\notag\\
&=\sum_{n,m,c}w_{mn}^cp_n\qty( s_c+\ln\frac{p_n}{p_m})\notag\\
&=\frac{1}{2}\sum_{n,m,c}j_{mn}^{c}\ln\frac{w_{mn}^cp_n}{w_{nm}^{c'}p_m}\notag\\
&=\frac{1}{2}\sum_{n,m,c}e^{-s_c/2}w_{mn}^c(f_{mn}^c)^2\Phi\qty(e^{s_c/2}p_n,e^{-s_c/2}p_m),\label{eq:ent.prod.rate}
\end{align}
where $\Phi(x,y)=(x-y)/\ln(x/y)$ is the logarithmic mean of $x$ and $y$.

\subsection{Proof}
We are now ready to prove the fidelity-dissipation relation \eqref{eq:main.res}. Noting that $U_g=U_\tau$ and $\tilde{\varrho}_0=\dyad{\varphi}$, we can calculate the fidelity between the final states of the ideal and actual quantum gates as follows:
\begin{align}
	\mca{F}_\varphi(\tau)&\coloneqq\mel{\varphi}{U_g^\dagger\varrho_\tau U_g}{\varphi}\notag\\
	&=\mel{\varphi}{\tilde{\varrho}_\tau}{\varphi}\notag\\
	&=1+\int_0^\tau\dd{t}\mel{\varphi}{\dot{\tilde{\varrho}}_t}{\varphi}.\label{eq:fidelity.int}
\end{align}
Note that $f_{mn}^c=-f_{nm}^{c'}$.
Expressing the initial pure state in terms of eigenbasis $\{\ket{n}\}$, i.e., $\ket{\varphi}=\sum_nz_n\ket{n}$, we can calculate the term in the integral of Eq.~\eqref{eq:fidelity.int} as
\begin{widetext}
\begin{align}
	\mel{\varphi}{\dot{\tilde{\varrho}}}{\varphi}&=\mel{\varphi}{\sum_c\mca{D}[\tilde{L}_c]\tilde{\varrho}}{\varphi}\notag\\
	&=\mel{\varphi}{\sum_c\qty[ \tilde{L}_c\tilde{\varrho}\tilde{L}_c^\dagger - \{\tilde{L}_c^\dagger \tilde{L}_c,\tilde{\varrho}\}/2]}{\varphi}\notag\\
	&=\sum_{n,m,k,c}z_m^*z_kp_n\mel{m}{\tilde{L}_c}{n}\mel{n}{\tilde{L}_c^\dagger}{k}-\frac{1}{2}z_m^*z_kp_k\mel{m}{\tilde{L}_c^\dagger}{n}\mel{n}{\tilde{L}_c}{k}-\frac{1}{2}z_m^*z_kp_m\mel{m}{\tilde{L}_c^\dagger}{n}\mel{n}{\tilde{L}_c}{k}\notag\\
	&=\sum_{n,m,k,c}z_m^*z_kp_n\mel{m}{\tilde{L}_c}{n}\mel{n}{\tilde{L}_c^\dagger}{k}-\frac{1}{2}e^{-s_c}z_m^*z_kp_k\mel{m}{\tilde{L}_c}{n}\mel{n}{\tilde{L}_c^\dagger}{k}-\frac{1}{2}e^{-s_c}z_m^*z_kp_m\mel{m}{\tilde{L}_c}{n}\mel{n}{\tilde{L}_c^\dagger}{k}\notag\\
	&=\frac{1}{2}\sum_{n,m,k,c}z_m^*z_k\mel{m}{\tilde{L}_c}{n}\mel{n}{\tilde{L}_c^\dagger}{k}\qty[2p_n-e^{-s_c}(p_k+p_m)]\notag\\
	&=\frac{1}{2}\sum_{n,m,k,c}e^{-s_c/2}z_m^*z_k\mel{m}{\tilde{L}_c}{n}\mel{n}{\tilde{L}_c^\dagger}{k}\qty[f_{kn}^c\Phi\qty(e^{s_c/2}p_n,e^{-s_c/2}p_k)+f_{mn}^c\Phi\qty(e^{s_c/2}p_n,e^{-s_c/2}p_m)]\notag\\
	&=\frac{1}{2}\sum_{n,m,k,c}f_{mn}^c\Phi\qty(e^{s_c/2}p_n,e^{-s_c/2}p_m)\qty[ e^{-s_c/2}z_m^*z_k\mel{m}{\tilde{L}_c}{n}\mel{n}{\tilde{L}_c^\dagger}{k} - e^{-s_{c'}/2}z_k^*z_n\mel{k}{\tilde{L}_{c'}}{m}\mel{m}{\tilde{L}_{c'}^\dagger}{n} ]\notag\\
	&=\frac{1}{2}\sum_{n,m,k,c}f_{mn}^c\Phi\qty(e^{s_c/2}p_n,e^{-s_c/2}p_m)\qty[ e^{-s_c/2}z_m^*z_k\mel{m}{\tilde{L}_c}{n}\mel{n}{\tilde{L}_c^\dagger}{k} - e^{-s_c/2} z_k^*z_n\mel{k}{\tilde{L}_c^\dagger}{m}\mel{m}{\tilde{L}_c}{n} ]\notag\\
	&=\frac{1}{2}\sum_{n,m,c}e^{-s_c/2}\mel{m}{\tilde{L}_c}{n}f_{mn}^c\Phi\qty(e^{s_c/2}p_n,e^{-s_c/2}p_m)\sum_k\qty[ z_m^*z_k\mel{n}{\tilde{L}_c^\dagger}{k} - z_k^*z_n\mel{k}{\tilde{L}_c^\dagger}{m} ]\notag\\
	&=\frac{1}{2}\sum_{n,m,c}e^{-s_c/2}\mel{m}{\tilde{L}_c}{n}f_{mn}^c\Phi\qty(e^{s_c/2}p_n,e^{-s_c/2}p_m)\mel{n}{[\tilde{L}_c^\dagger,\dyad{\varphi}]}{m}.\label{eq:fidelity.tmp0}
\end{align}
Here, we perform the index swap $k\to n$, $n\to m$, $m\to k$, and $c\to c'$ for the term
\begin{equation}
	e^{-s_c/2}z_m^*z_k\mel{m}{\tilde{L}_c}{n}\mel{n}{\tilde{L}_c^\dagger}{k}f_{kn}^c\Phi\qty(e^{s_c/2}p_n,e^{-s_c/2}p_k)
\end{equation}
in the sixth line in to obtain the seventh line.
Applying the Cauchy-Schwarz inequality to Eq.~\eqref{eq:fidelity.tmp0}, we obtain
\begin{align}
	\qty|\mel{\varphi}{\dot{\tilde{\varrho}}}{\varphi}|^2&\le\frac{1}{4}\sum_{n,m,c}e^{-s_c/2}w_{mn}^c(f_{mn}^c)^2\Phi\qty(e^{s_c/2}p_n,e^{-s_c/2}p_m)\sum_{n,m,c}e^{-s_c/2}\Phi\qty(e^{s_c/2}p_n,e^{-s_c/2}p_m)\qty|\mel{n}{[\tilde{L}_c^\dagger,\dyad{\varphi}]}{m}|^2\notag\\
	&=\frac{1}{2}\dot\Sigma_\varphi\times\sum_{n,m,c}e^{-s_c/2}\Phi\qty(e^{s_c/2}p_n,e^{-s_c/2}p_m)\qty| \mel{n}{[\tilde{L}_c^\dagger,\dyad{\varphi}]}{m} |^2.\label{eq:fidelity.tmp1}
\end{align}
By applying the inequalities $\Phi(x,y)\le(x+y)/2$ and $|\tr{AB}|\le\|A\|_\infty\|B\|_1$, the second term in Eq.~\eqref{eq:fidelity.tmp1} can be upper bounded as
\begin{align}
	&\sum_{n,m,c}e^{-s_c/2}\Phi\qty(e^{s_c/2}p_n,e^{-s_c/2}p_m)\qty| \mel{n}{[\tilde{L}_c^\dagger,\dyad{\varphi}]}{m} |^2\notag\\
	&\le \frac{1}{2}\sum_{n,m,c}e^{-s_c/2}\qty(e^{s_c/2}p_n+e^{-s_c/2}p_m)\qty| \mel{n}{[\tilde{L}_c^\dagger,\dyad{\varphi}]}{m} |^2\notag\\
	&=\frac{1}{2}\sum_{n,m,c}p_n\qty| \mel{n}{[\tilde{L}_c^\dagger,\dyad{\varphi}]}{m} |^2+\frac{1}{2}\sum_{n,m,c}p_m\qty| \mel{n}{[\tilde{L}_{c'},\dyad{\varphi}]}{m} |^2\notag\\
	&=\frac{1}{2}\sum_{c}\tr{[\tilde{L}_c^\dagger,\dyad{\varphi}][\tilde{L}_c^\dagger,\dyad{\varphi}]^\dagger\tilde\varrho}+\frac{1}{2}\sum_{c}\tr{[\tilde{L}_{c'},\dyad{\varphi}]^\dagger[\tilde{L}_{c'},\dyad{\varphi}]\tilde\varrho}\notag\\
	&\le\frac{1}{2}\sum_{c}\|[\tilde{L}_c^\dagger,\dyad{\varphi}][\tilde{L}_c^\dagger,\dyad{\varphi}]^\dagger\|_\infty+\frac{1}{2}\sum_{c}\|[\tilde{L}_{c'},\dyad{\varphi}]^\dagger[\tilde{L}_{c'},\dyad{\varphi}]\|_\infty\notag\\
	&=\sum_{c}\|[\tilde{L}_c,\dyad{\varphi}]\|_\infty^2\notag\\
	&=\sum_{c}(\delta_\varphi \tilde{L}_c)^2,\label{eq:fidelity.tmp2}
\end{align}
\end{widetext}
where we define $\delta_\varphi L\coloneqq\sqrt{\mel{\varphi}{L^\dagger L}{\varphi}-|\mel{\varphi}{L}{\varphi}|^2}$ and apply Lemma \ref{lem:op.norm} to obtain the last line.
By combining Eqs.~\eqref{eq:fidelity.tmp1} and \eqref{eq:fidelity.tmp2}, we arrive at the following inequality:
\begin{align}
	[1-\mca{F}_\varphi(\tau)]^2&=\qty(\int_0^\tau\dd{t}\mel{\varphi}{\dot{\tilde{\varrho}}_t}{\varphi})^2\notag\\
	&\le \qty(\int_0^\tau\dd{t}\sqrt{\dot\Sigma_\varphi(t)\sum_c[\delta_\varphi \tilde{L}_c(t)]^2/2})^2\notag\\
	&\le \frac{1}{2}\int_0^\tau\dd{t}\sum_c[\delta_\varphi \tilde{L}_c(t)]^2\int_0^\tau\dd{t}\dot\Sigma_\varphi(t)\notag\\
	&= \gamma_\varphi\Sigma_\varphi(\tau)/2,
\end{align}
where we define $\gamma_\varphi\coloneqq\int_0^\tau\dd{t}\sum_c[\delta_\varphi \tilde{L}_c(t)]^2$.
Consequently, we obtain the fidelity-dissipation relation for an initial pure state as follows:
\begin{equation}
	\mca{F}_\varphi(\tau)+\sqrt{\gamma_\varphi\Sigma_\varphi(\tau)/2}\ge 1.
\end{equation}
Taking the ensemble average over all pure states, we obtain the desired relation \eqref{eq:main.res} as
\begin{align}
	1&\le \int\dd{\varphi}\qty[\mca{F}_\varphi(\tau)+\sqrt{\gamma_\varphi\Sigma_\varphi(\tau)/2}]\notag\\
	&\le {\mca{F}} + \sqrt{(1/2)\int\dd{\varphi}\gamma_\varphi\int\dd{\varphi}\Sigma_\varphi(\tau)}\notag\\
	&= {\mca{F}} + \sqrt{\gamma{\Sigma}/2},
\end{align}
where $\gamma$ is defined as $\gamma\coloneqq\int\dd{\varphi}\gamma_\varphi$.
Using Lemma \ref{lem:delta.L}, $\gamma$ can be explicitly calculated as
\begin{align}
	\gamma&=\frac{1}{d(d+1)}\sum_c\int_0^\tau\dd{t}\qty(d\tr{\tilde{L}_c(t)^\dagger \tilde{L}_c(t)} - |\tr \tilde{L}_c(t)|^2)\notag\\
	&=\frac{1}{d(d+1)}\sum_c\int_0^\tau\dd{t}\qty(d\tr{L_c(t)^\dagger L_c(t)} - |\tr L_c(t)|^2)\notag\\
	&=\int_0^\tau\dd{t}\sum_c[\Delta L_c(t)]^2,
\end{align}
where $[\Delta L]^2\coloneqq(d\tr{L^\dagger L} - |\tr L|^2)/d(d+1)\ge 0$.

\begin{lemma}\label{lem:op.norm}
For an arbitrary matrix $L$ and pure state $\ket{\varphi}$, the following relation always holds:
\begin{equation}\label{eq:lem1}
	\|[L,\dyad{\varphi}]\|_\infty=\sqrt{\mel{\varphi}{L^\dagger L}{\varphi}-|\mel{\varphi}{L}{\varphi}|^2}\eqqcolon \delta_\varphi L\le\|L\|_\infty.
\end{equation}
\end{lemma}
\begin{proof}
Defining the normalized pure state
\begin{equation}
	\ket{\varphi_\perp}\coloneqq\frac{(L-\mel{\varphi}{L}{\varphi})\ket{\varphi}}{\delta_\varphi L},
\end{equation}
we can easily confirm that $\braket{\varphi}{\varphi_\perp}=0$ and
\begin{equation}
	[L,\dyad{\varphi}]^\dagger [L,\dyad{\varphi}]=(\delta_\varphi L)^2\qty(\dyad{\varphi}+\dyad{\varphi_\perp}).
\end{equation}
Denoting the maximum eigenvalue of a matrix $A$ as $\lambda_{\rm max}(A)$, Eq.~\eqref{eq:lem1} can be proved as follows:
\begin{align}
	\|[L,\dyad{\varphi}]\|_\infty&=\lambda_{\rm max}(|[L,\dyad{\varphi}]|)\notag\\
	&=\lambda_{\rm max}(\sqrt{[L,\dyad{\varphi}]^\dagger [L,\dyad{\varphi}]})\notag\\
	&=\lambda_{\rm max}(\delta_\varphi L\sqrt{\dyad{\varphi}+\dyad{\varphi_\perp}})\notag\\
	&=\delta_\varphi L\notag\\
	&\le\sqrt{\mel{\varphi}{L^\dagger L}{\varphi}}\notag\\
	&\le\|L\|_\infty.
\end{align}
\end{proof}

\begin{lemma}\label{lem:delta.L}
For an arbitrary matrix $L$, we have
\begin{equation}
\int\dd{\varphi}(\delta_\varphi L)^2=\frac{d\tr{L^\dagger L} - |\tr L|^2}{d(d+1)}.
\end{equation}
\end{lemma}
\begin{proof}
We note that the following relation holds for arbitrary matrices $X$ and $Y$ \cite{Dankert.2005.arxiv}:
\begin{equation}\label{eq:ens.rel}
\int\dd{\varphi}\mel{\varphi}{X}{\varphi}\mel{\varphi}{Y}{\varphi}=\frac{1}{d(d+1)}\qty( \tr{XY}+\tr X\tr Y ).
\end{equation}
Substituting $X=L^\dagger L$ and $Y=\mds{1}$ to Eq.~\eqref{eq:ens.rel}, we obtain
\begin{equation}
	\int\dd{\varphi}\mel{\varphi}{L^\dagger L}{\varphi}=\frac{\tr{L^\dagger L}}{d}.
\end{equation}
Similarly, substituting $X=L$ and $Y=L^\dagger$ to Eq.~\eqref{eq:ens.rel}, we obtain
\begin{equation}
	\int\dd{\varphi}|\mel{\varphi}{L}{\varphi}|^2=\frac{1}{d(d+1)}\qty(\tr{L^\dagger L}+|\tr L|^2).
\end{equation}
Using these relations, we can calculate $\int\dd{\varphi}(\delta_\varphi L)^2$ as follows:
\begin{align}
	\int\dd{\varphi}(\delta_\varphi L)^2&=\int\dd{\varphi}\qty(\mel{\varphi}{L^\dagger L}{\varphi}-|\mel{\varphi}{L}{\varphi}|^2)\notag\\
	&=\frac{\tr{L^\dagger L}}{d}-\frac{1}{d(d+1)}\qty(\tr{L^\dagger L}+|\tr L|^2)\notag\\
	&=\frac{d\tr{L^\dagger L} - |\tr L|^2}{d(d+1)}.
\end{align}
\end{proof}

\section{Proof of the generalized relation \eqref{eq:main.res.gen} for the imperfect-Hamiltonian case}\label{app:proof.res.gen}
We consider the case that the implemented Hamiltonian $\hat H_t$ is different from the ideal Hamiltonian $H_t$.
We follow the same strategy in the previous section by defining the unitary operator in the interaction picture $\hat U_t\coloneqq\vec{\mbb{T}}\exp\qty{-i\int_0^t\dd{s}\hat H_s}$.
Note that $\hat U_\tau\neq U_g$.
In this case, the fidelity can be calculated as
\begin{align}
	\mca{F}_\varphi(\tau)&=\mel{\varphi}{U_g^\dagger\varrho_\tau U_g}{\varphi}=\mel{\varphi}{U_g^\dagger\hat U_\tau\tilde\varrho_\tau\hat U_\tau^\dagger U_g}{\varphi}\notag\\
	&=|\braket{\varphi'}{\varphi}|^2+\int_0^\tau\dd{t}\mel{\varphi'}{\dot{\tilde{\varrho}}_t}{\varphi'},
\end{align}
where we define $\ket{\varphi'}\coloneqq \hat U_\tau^\dagger U_g\ket{\varphi}$.
Following the same procedure in the perfect case, we obtain the generalized relation for the initial pure state $\dyad{\varphi}$, given by
\begin{equation}
	\mca{F}_\varphi(\tau)+\sqrt{\gamma_{\varphi'}\Sigma_\varphi(\tau)/2}\ge|\mel{\varphi}{U_g^\dagger\hat U_\tau}{\varphi}|^2.
\end{equation}
Taking the ensemble average over all pure states and noticing that
\begin{align}
	\int\dd{\varphi}|\mel{\varphi}{U_g^\dagger\hat U_\tau}{\varphi}|^2&=\frac{|\tr{\hat U_\tau^\dagger U_g}|^2+d}{d(d+1)},\\
	\int\dd{\varphi}\gamma_{\varphi'}&=\gamma,
\end{align}
we immediately obtain the generalized fidelity-dissipation relation \eqref{eq:main.res.gen},
\begin{equation}
	\mca{F}+\sqrt{\gamma\Sigma/2}\ge\frac{|\tr{\hat U_\tau^\dagger U_g}|^2+d}{d(d+1)}.
\end{equation}

\section{Derivation of the fidelity-activity trade-off relation}\label{app:proof.fa}
Here we derive a trade-off relation between fidelity and activity for arbitrary quantum gates without the assumption of local detailed balance condition.
To this end, we calculate $\mel{\varphi}{\dot{\tilde{\varrho}}}{\varphi}$ as follows:
\begin{align}
	\mel{\varphi}{\dot{\tilde{\varrho}}}{\varphi}&=\mel{\varphi}{\sum_c\mca{D}[\tilde{L}_c]\tilde{\varrho}}{\varphi}\notag\\
	&=\mel{\varphi}{\sum_c\qty[ \tilde{L}_c\tilde{\varrho}\tilde{L}_c^\dagger - \{\tilde{L}_c^\dagger \tilde{L}_c,\tilde{\varrho}\}/2]}{\varphi}\notag\\
	&=\frac{1}{2}\sum_c\tr{(\tilde{L}_c^\dagger[\dyad{\varphi},\tilde{L}_c]+[\tilde{L}_c^\dagger,\dyad{\varphi}]\tilde{L}_c)\tilde{\varrho}}.\label{eq:fa.tmp0}
\end{align}
Applying the inequalities $|\tr{AB}|\le\|A\|_\infty\|B\|_1$, $\|AB\|_\infty\le\|A\|_\infty\|B\|_\infty$, and $\|A+B\|_\infty\le\|A\|_\infty+\|B\|_\infty$, we can bound from above as
\begin{align}
	|\mel{\varphi}{\dot{\tilde{\varrho}}}{\varphi}|&\le\frac{1}{2}\sum_c\|\tilde{L}_c\|_\infty(\delta_\varphi\tilde{L}_c+\delta_\varphi\tilde{L}_c^\dagger).
\end{align}
Consequently, we obtain
\begin{equation}\label{eq:fa.tmp1}
	1-\mca{F}_\varphi(\tau)\le\frac{1}{2}\int_0^\tau\dd{t}\sum_c\|\tilde{L}_c(t)\|_\infty[\delta_\varphi\tilde{L}_c(t)+\delta_\varphi\tilde{L}_c(t)^\dagger].
\end{equation}
Applying the Cauchy-Schwarz inequality, we can prove that
\begin{align}
	\int\dd{\varphi}\delta_\varphi\tilde{L}_c(t)&\le\sqrt{\int\dd{\varphi}[\delta_\varphi\tilde{L}_c(t)]^2}=\Delta L_c(t).\label{eq:fa.tmp3}
\end{align}
Note that $\Delta L=\Delta L^\dagger$.
Taking the ensemble average in Eq.~\eqref{eq:fa.tmp1} and applying Eq.~\eqref{eq:fa.tmp3}, we get
\begin{align}
	1-\mca{F}&\le\frac{1}{2}\int_0^\tau\dd{t}\sum_c\|\tilde{L}_c(t)\|_\infty\int\dd{\varphi}[\delta_\varphi\tilde{L}_c(t)+\delta_\varphi\tilde{L}_c(t)^\dagger]\notag\\
	&\le\int_0^\tau\dd{t}\sum_c\|L_c(t)\|_\infty\Delta L_c(t)\eqqcolon\upsilon.\label{eq:fa.tmp2}
\end{align}
Here we define the activity quantity $\upsilon$, which is determined solely by the norm of jump operators.
Physically, $\upsilon$ quantifies the thermalization strength of the quantum gate. 
By rearranging Eq.~\eqref{eq:fa.tmp2}, the trade-off relation between fidelity and activity is immediately achieved,
\begin{equation}
	\mca{F}+\upsilon\ge 1.
\end{equation}

\section{Proof of the relation \eqref{eq:main.res.therm.relax} for time-independent quantum gates}\label{app:proof.res.ther.rel}
For an open quantum system that undergoes a thermal relaxation process, it was shown that the irreversible entropy production is always lower bounded by a relative entropy as \cite{Vu.2021.PRL2}
\begin{equation}\label{eq:ent.prod.lb}
	\Sigma_\varphi(\tau)\ge S_{\rm M}(U_g\varrho_0 U_g^\dagger||\varrho_\tau).
\end{equation}
Here, $S_{\rm M}$ is the projectively measurement relative entropy, given by
\begin{equation}
	S_{\rm M}(\varrho||\sigma)\coloneqq S(\varrho||\sum_n\Pi_n\sigma\Pi_n),
\end{equation}
where $\{\Pi_n\}$ are the eigenstates of $\varrho$.
Since the initial state is a pure state (i.e., $\varrho_0=\dyad{\varphi}$), Eq.~\eqref{eq:ent.prod.lb} can be proceeded further to obtain
\begin{equation}
	\Sigma_\varphi(\tau)\ge -\ln\mel{\varphi}{U_g^\dagger\varrho_\tau U_g}{\varphi}=-\ln\mca{F}_\varphi(\tau),
\end{equation}
or equivalently $\mca{F}_\varphi(\tau)\ge e^{-\Sigma_\varphi(\tau)}$.
Since $e^{-x}$ is the convex function, by applying Jensen's inequality, the relation \eqref{eq:main.res.therm.relax} can be derived as follows:
\begin{equation}
	\mca{F}\ge\int\dd{\varphi}e^{-\Sigma_\varphi(\tau)}\ge e^{-\int\dd{\varphi}\Sigma_\varphi(\tau)}=e^{-\Sigma}.
\end{equation}

\section{Proof of the trade-off relation \eqref{eq:FQ.tradeoff} for general time-independent gates}\label{app:proof.ener.diss}
For convenience, we define the inner product $\ev{X,Y}\coloneqq\tr{X^\dagger Y}$ and the Frobenius norm $\|X\|_F^2\coloneqq\ev{X,X}$.
First, we derive the Kraus representation for the quantum gate.
Let $\varrho_E=\sum_\mu\lambda_\mu\dyad{\lambda_\mu}$ be the spectral decomposition of the initial state of the environment.
Using this expression, the Kraus representation for the quantum state of the gate at time $\tau$ can be obtained as
\begin{align}
	\varrho_\tau&=\tr_E\{U(\varrho_0\otimes\varrho_E)U^\dagger\}\notag\\
	&=\sum_{\mu}\bra{\lambda_\mu}U(\varrho_0\otimes\varrho_E)U^\dagger\ket{\lambda_\mu}\notag\\
	&=\sum_{\mu,\nu}\bra{\lambda_\mu}U(\varrho_0\otimes \lambda_\nu\dyad{\lambda_\nu})U^\dagger\ket{\lambda_\mu}\notag\\
	&=\sum_{\mu,\nu}K_{\mu\nu}\varrho_0K_{\mu\nu}^\dagger,
\end{align}
where we define $K_{\mu\nu}\coloneqq\sqrt{\lambda_\nu}\mel{\lambda_\mu}{U}{\lambda_\nu}$.
It can be easily confirmed that $\sum_{\mu,\nu}K_{\mu\nu}^\dagger K_{\mu\nu}=\mds{1}$.
Using this Kraus representation, the average gate fidelity can be analytically calculated as
\begin{align}
	\mca{F}&=\int\dd{\varphi}\mel{\varphi}{U_g^\dagger\varrho_\tau U_g}{\varphi}\notag\\
	&=\sum_{\mu,\nu}\int\dd{\varphi}\mel{\varphi}{U_g^\dagger K_{\mu\nu}\dyad{\varphi}K_{\mu\nu}^\dagger U_g}{\varphi}\notag\\
	&=\frac{1}{d(d+1)}\sum_{\mu,\nu}\qty[ \tr{U_g^\dagger K_{\mu\nu}K_{\mu\nu}^\dagger U_g} + |\tr{U_g^\dagger K_{\mu\nu}}|^2]\notag\\
	&=\frac{1}{d(d+1)}\qty(d+\sum_{\mu,\nu}|\ev{U_g,K_{\mu\nu}}|^2).
\end{align}
Here we use the relation $\sum_{\mu,\nu}K_{\mu\nu}^\dagger K_{\mu\nu}=\mds{1}$ to obtain the last line.
Note that
\begin{equation}
	\ev{X,X}\ev{Y,Y}-|\ev{X,Y}|^2=\|X\|_F^2\left\|Y-\frac{\ev{Y,X}}{\|X\|_F^2}X\right\|_F^2.
\end{equation}
Using this relation, $1-\mca{F}$ can be calculated further as
\begin{align}
	1-\mca{F}&=\frac{1}{d(d+1)}\qty(d^2-\sum_{\mu,\nu}|\ev{U_g,K_{\mu\nu}}|^2)\notag\\
	&=\frac{1}{d(d+1)}\sum_{\mu,\nu}\qty[ \ev{U_g,U_g} \ev{K_{\mu\nu},K_{\mu\nu}} - |\ev{U_g,K_{\mu\nu}}|^2]\notag\\
	&=\frac{1}{d+1}\sum_{\mu,\nu}\left\| K_{\mu\nu} - \kappa_{\mu\nu}U_g \right\|_F^2\notag\\
	&=\frac{1}{d+1}\sum_{\mu,\nu}\left\| \bar{K}_{\mu\nu} \right\|_F^2,
\end{align}
where we define $\kappa_{\mu\nu}\coloneqq\ev{K_{\mu\nu},U_g}/d$ and $\bar{K}_{\mu\nu}\coloneqq K_{\mu\nu} - \kappa_{\mu\nu}U_g$.
Noting that $[U_g,H]=0$, $[U_g^\dagger,H]=0$, and $\ev{X,[Y,Z]}=\ev{Z^\dagger,[X^\dagger,Y]}$, the energy dissipation can also be calculated as follows:
\begin{align}
	\mca{Q}&=\int\dd{\varphi}\tr{H(\varrho_\tau-\varrho_0)}\notag\\
	&=\sum_{\mu,\nu}\int\dd{\varphi}\tr{HK_{\mu\nu}\dyad{\varphi}K_{\mu\nu}^\dagger}-\int\dd{\varphi}\mel{\varphi}{H}{\varphi}\notag\\
	&=\frac{1}{d}\qty(\sum_{\mu,\nu}\tr{K_{\mu\nu}^\dagger HK_{\mu\nu}}-\tr H)\notag\\
	&=\frac{1}{d}\sum_{\mu,\nu}\tr{K_{\mu\nu}^\dagger HK_{\mu\nu}-HK_{\mu\nu}^\dagger K_{\mu\nu}}\notag\\
	&=\frac{1}{d}\sum_{\mu,\nu}\ev{K_{\mu\nu},[H,K_{\mu\nu}]}\notag\\
	&=\frac{1}{d}\sum_{\mu,\nu}\ev{K_{\mu\nu}-\kappa_{\mu\nu}U_g,[H-\epsilon_*\mds{1},K_{\mu\nu}-\kappa_{\mu\nu}U_g]}\notag\\
	&=\frac{1}{d}\sum_{\mu,\nu}\ev{\bar{K}_{\mu\nu},[H-\epsilon_*\mds{1},\bar{K}_{\mu\nu}]}.
\end{align}
Here, $\epsilon_*$ is an arbitrary number that will be determined later.
Applying Lemma \ref{lem:F.norm} to the energy dissipation term, we can upper bound $|\mca{Q}|$ as
\begin{align}
	d|\mca{Q}|&\le 2\|H-\epsilon_*\mds{1}\|_\infty\sum_{\mu,\nu}\|\bar{K}_{\mu\nu}\|_F^2\notag\\
	&=2\|H-\epsilon_*\mds{1}\|_\infty(d+1)(1-\mca{F}).\label{eq:fqr.tmp1}
\end{align}
Let $H=\sum_n\epsilon_n\dyad{\epsilon_n}$ be the spectral decomposition of the gate Hamiltonian.
Choosing $\epsilon_*=(\max_n\epsilon_n+\min_n\epsilon_n)/2$, we have $\|H-\epsilon_*\mds{1}\|_\infty=g/2$, where $g\coloneqq\max_n\epsilon_n-\min_n\epsilon_n$.
Consequently, by rearranging Eq.~\eqref{eq:fqr.tmp1}, we readily obtain the desired trade-off relation \eqref{eq:FQ.tradeoff} in the main text as
\begin{equation}
	\mca{F}+\frac{d}{(d+1)g}|\mca{Q}|\le 1.
\end{equation}

\subsection{Generalization to general time-dependent gates}
Here we derive a generalization of the trade-off relation \eqref{eq:FQ.tradeoff} to the case of time-dependent Hamiltonians.
In this case, the system energy can change due to the work injected from the external control even when the quantum gate is not affected by the environment. In other words, fidelity can be at its maximum value of 1, whereas the energy change is nonzero. Therefore, there is no trade-off relation between fidelity and energy change with respect to the system Hamiltonian $H$. 
To establish the trade-off relation in this situation, we consider an effective Hamiltonian $\bar{H}$ that yields the same unitary evolution,
\begin{equation}
	e^{-i\bar{H}\tau}=U_g=\vec{\mbb{T}}\exp(-i\int_0^\tau\dd{t}H_t).
\end{equation}
Note that $[\bar{H},U_g]=0$ and $\bar{H}^\dagger=\bar{H}$.
The average energy change with respect to this effective Hamiltonian after applying the gate can be defined as
\begin{equation}
	\bar{\mca{Q}}\coloneqq\int\dd{\varphi}\tr{\bar{H}(\varrho_\tau-\varrho_0)}.
\end{equation}
$\bar{\mca{Q}}$ can be considered as the average energy change in the effective dynamics.
Evidently, $\bar{\mca{Q}}$ vanishes if the quantum gate is not affected by the environment.
Then, following the same procedure as in the time-independent case, we immediately obtain the following trade-off relation:
\begin{equation}
	\mca{F}+\frac{d}{(d+1)\bar{g}}|\bar{\mca{Q}}|\le 1.
\end{equation}
Here, $\bar{g}$ denotes the energy bandwidth of the effective Hamiltonian $\bar{H}$ (i.e., $\bar{g}\coloneqq\max_n\bar{\epsilon}_n-\min_n\bar{\epsilon}_n$, where $\{\bar{\epsilon}_n\}$ are the eigenvalues of $\bar{H}$).

\begin{lemma}\label{lem:F.norm}
For arbitrary matrices $X$ and $Y$, the following relation always holds:
\begin{equation}\label{eq:lem3}
	|\ev{X,[Y,X]}|\le 2\|Y\|_\infty\|X\|_F^2.
\end{equation}
\end{lemma}
\begin{proof}
Applying inequality $|\tr{YX^\dagger X}|\le\|Y\|_\infty\|X^\dagger X\|_1=\|Y\|_\infty\|X\|_F^2$, Eq.~\eqref{eq:lem3} can be proved as follows:
\begin{align}
	|\ev{X,[Y,X]}|&\le|\tr{YXX^\dagger}|+|\tr{YX^\dagger X}|\notag\\
	&\le 2\|Y\|_\infty\|X\|_F^2.
\end{align}
\end{proof}


\begin{thebibliography}{84}
\bibitem [{fnt()}]{fnt1}%
  \BibitemOpen
  \href@noop {} {}We note that the quantity $\gamma$ in our result involves the
  jump operators $\{L_k\}$. However, $\gamma$ depends only on their norms,
  which are easier to estimate than the exact forms of the jump
  operators.\BibitemShut {Stop}%
\end{thebibliography}

%

\end{document}